\documentclass[11pt]{article}
\usepackage{fullpage,hyperref,bbm,amssymb,amsfonts,amsthm,amsmath}
\usepackage{palatino}
\usepackage{graphicx}

\newcommand{\PPPP}{{\textsf{P}}}
\newcommand{\NP}{{\textsf{NP}}}
\newcommand{\BPP}{{\textsf{BPP}}}
\newcommand{\MA}{{\textsf{MA}}}
\newcommand{\BQP}{{\textsf{BQP}}}
\newcommand{\QMA}{{\textsf{QMA}}}

\renewcommand{\>}{\rangle}

\newcommand{\cH}{\mathcal{H}}

\newcommand{\ket}[1]{\left| #1\right\rangle}      
\newcommand{\bra}[1]{\left\langle #1\right|}      
\newcommand{\kets}[1]{| #1 \rangle}                 
\newcommand{\bras}[1]{\langle #1 |}                 
\newcommand{\braket}[2]{\langle #1 | #2 \rangle}         
\newcommand{\ii}{\mathbb{I}} 
\newcommand{\norm}[1]{\left\| #1\right\|}                
\newcommand{\ep}{\epsilon} 
\newtheorem{definition}{Definition}
\newtheorem{theorem}{Theorem}
\newtheorem*{theorem3}{Theorem 3'}
\newtheorem{lemma}{Lemma}

  \newcommand{\eq}{\begin{equation}}
  \newcommand{\en}{\end{equation}}
  \newcommand{\eqa}{\begin{eqnarray}}
  \newcommand{\ena}{\end{eqnarray}}

\begin{document}
    \title{Fast Amplification of QMA}
        \author{Daniel Nagaj\thanks{Research Center for Quantum Information, Institute of Physics, Slovak
        Academy of Sciences, D\'ubravsk\'a cesta 9, 84215 Bratislava, Slovakia,
        and Quniverse, L\'{i}\v{s}\v{c}ie \'{u}dolie 116, 841 04, Bratislava, Slovakia.
        Email: \texttt{daniel.nagaj@savba.sk}},
        \quad
        Pawel Wocjan\thanks{School of Electrical Engineering and Computer
        Science, University of Central Florida, Orlando, FL~32816, USA. Email:
        \texttt{wocjan@eecs.ucf.edu}},
        \quad
        Yong Zhang\thanks{School of Electrical Engineering
        and Computer Science, University of Central Florida, Orlando, FL~32816, USA. Email:
        \texttt{yong@cs.ucf.edu}}
        }
        \date{\today}
    \maketitle

\begin{abstract}
Given a verifier circuit for a problem in $\QMA$, we show how to exponentially amplify 
the gap between its acceptance probabilities in the `yes' and `no' cases,
with a method that is quadratically faster than the procedure given
by Marriott and Watrous \cite{MWqma}.
Our construction is natively quantum, based on the analogy of 
a product of two reflections and a quantum walk.
Second, in some special cases we show how to amplify the acceptance probability 
for good witnesses to 1, making a step towards the proof that 
$\QMA$ with one-sided error ($\QMA_1$) is equal to $\QMA$. 
Finally, we simplify the filter-state method to search for $\QMA$ witnesses
by Poulin and Wocjan \cite{PoulinWocjan}.
\end{abstract}


\section{Introduction}

Which decision problems (yes/no questions) can be efficiently 
solved on a classical computer? All such problems constitute the complexity class $\PPPP$. The goal of many algorithm designers is to place problems into $\PPPP$ by finding efficient algorithms for them\footnote{See e.g. the recent example of an algorithm for deciding whether a number is prime \cite{primesinP}.}.
Another notable complexity class called $\NP$, is the class of problems whose solutions can be efficiently verified. However, finding these solutions might be hard, as $\NP$ contains notoriously hard problems like \textsc{Satisfiability} \cite{SATreference}. Whether all the problems in $\NP$ could be actually solved in polynomial time (i.e. $\PPPP \stackrel{?}{=} \NP$) is one of the most interesting open questions of modern computer science \cite{PvsNP}. 

The problems efficiently solvable with randomized circuits, i.e. with circuits that are allowed to fail with some bounded probability, constitute the class $\BPP$. When we now allow the solution-verifying procedure in the definition of $\NP$ to have a small probability of failure, we get the complexity class $\MA$. This acronym stands for ``Merlin-Arthur'', as the verifying protocol for the problems in $\MA$ goes like this: an all-powerful Merlin provides a `proof' (also called a witness) which a rational Arthur verifies. 
The problems in $\MA$ are those for which Merlin can convince Arthur that the answer to his question is `yes' if it is so, while he has a low probability of fooling Arthur in the `no' cases. 

The world is quantum mechanical, so it is natural to ask what can be efficiently computed on a quantum computer? All such problems form the complexity class $\BQP$, and include such problems as factoring \cite{Shor} and approximating the Jones Polynomial \cite{WocjanYard}.
Kitaev \cite{KitaevBook} and Watrous \cite{WatrousQMA} defined a quantum analogue of the class $\MA$, calling it $\QMA$ (Quantum Merlin-Arthur). 

Let us look at the verifying procedure in more detail, starting with an exact definition of $\QMA$.
\begin{definition}[QMA]
\label{QMAdef}
	Consider a language (a set of `yes'/`no' questions) $L=L_{yes}\cup L_{no}$. Denote its instances $x$. 
	The language $L$ belongs to the class $\QMA$ if 
	\begin{enumerate}
		\item	there exists a uniform family of quantum verifier circuits $V$
					working on $n=poly(|x|)$ qubits and $m=poly(|x|)$ ancillae,
					and two numbers $a, b$, with separation	$a-b$ 
					lower bounded by an inverse polynomial in $n$,
		\item for $x\in L_{yes}$ (the answer to the question $x$ is `yes'), 
					there exists a witness $\ket{\psi}$ such that the circuit $V$ on 
					$\ket{\psi}\ket{0}^{\otimes m}$ outputs `yes' (Arthur is convinced) with probability 			
					$p_{max} \geq a$,
		\item for $x\in L_{no}$, for any state $\ket{\varphi}$, 
					the circuit $V$ on $\ket{\varphi}\ket{0}^{\otimes m}$ outputs `yes' 
					(is fooled) with probability $p_{max} \leq b$,
	\end{enumerate}
\end{definition}

When the separation $a-b$ for a verifier circuit $V$ is small, it could take Merlin many verification rounds to convince Arthur that the answer to the question $x$ really is `yes'. However, the separation $a-b$ can be {\em amplified} by modifying the original verifier circuit, obtaining a new amplified circuit
with strong promise bounds 
\begin{eqnarray}
	x\in L_{yes}: \ && p_{max} \geq 1-e^{-r}, \label{strongpromise}\\
	x\in L_{no}: \ && p_{max} \leq e^{-r}, \nonumber
\end{eqnarray}
for some constant $r$. 

The first such amplification procedure by Kitaev \cite{KitaevBook} uses a circuit made from 
\begin{eqnarray}
	N = \frac{c r}{(a-b)^2}
	\label{Nkit}
\end{eqnarray}
(for some constant $c$)
parallel copies of the circuit $V$, with majority voting at the end.
Kitaev showed\footnote{Note that it was necessary to show that entangled witnesses wouldn't help Merlin to cheat in the `no' cases.} that this new circuit has the strong promise bounds \eqref{strongpromise}. The drawback of this procedure is that Merlin now has to provide an $N$ times longer witness.
 
In \cite{MWqma}, Marriott and Watrous showed that a single quantum witness of length $n$ suffices, finding a verification procedure which reuses this witness many times. Their amplified circuit 
uses $N$ copies \eqref{Nkit} of the original circuit $V$ and its conjugate $V^{\dagger}$,
interspersed with some additional operations. The procedure ends with $N$ measurements
and {\em classical processing} of the results. Arguments employing Chernoff bounds are 
then used to establish the strong bounds \eqref{strongpromise} in the `yes' and `no' cases. 
We give further details about their approach in Section \ref{MWsection}. 

The first result of our paper is a new and faster $\QMA$ amplification method.
\begin{theorem}[Fast $\QMA$ Amplification]
	\label{amptheorem}
	Consider a verifier circuit $V$ for a problem in $\QMA$,
	acting on $n$ qubits and $m$ ancillae, with promise bounds $a$ and $b$.
	There exists an amplified verifier circuit $V'$ 
	acting on $n$ qubits and $poly(n)$ ancillae, using
	\begin{eqnarray}
		N' = \frac{c'r}{a-b}
	\end{eqnarray}
	evaluations of the original circuit $V$ and its conjugate $V^\dagger$, 
	with promise bounds amplified to \eqref{strongpromise}.
\end{theorem}

Note that the dependence of $N'$ on $\frac{1}{a-b}$ is quadratically better than in \eqref{Nkit}.
This speedup is possible because our circuit $V'$ uses intrinsically quantum methods for producing its final answer. There are two ideas behind our construction.
First, we utilize the connection between a product of two reflections and a generalized quantum walk.
Second, the final {\em coherent processing} in the circuit is based on phase estimation \cite{phaseestimate} 
(or alternatively, the filter state method of Poulin and Wocjan \cite{PoulinWocjan})
instead of classical majority voting (as in \cite{KitaevBook}) or counting (as in \cite{MWqma}).

When $a$ in the definition of $\QMA$ is exactly 1, we get a special case of $\QMA$, called $\QMA$$_1$,
or $\QMA$ with one-sided error. In this case, there exists a witness which the circuit $V$
accepts without failing. A typical example of a problem in this class is Bravyi's Quantum $k$-SAT \cite{BravyiQSAT}.
In the classical world, the class $\MA$ with probabilistic verifier circuits
can be amplified to $\MA$ without errors \cite{MAerror1,MAerror2}. However, it is a big open question, whether the corresponding quantum classes are equal, i.e. whether $\QMA=\QMA_1$. In \cite{Aqma1}, Aaronson argued that this problem is hard and gave an oracle separating them. 
The second result of our paper is a step towards answering whether $\QMA\stackrel{?}{=}\QMA_1$.
We show that for a class of $\QMA$ verifier circuits, we can amplify the probability promise $a$ up to 1 exactly, as opposed to exponentially approaching it as in \eqref{strongpromise}. 

\begin{theorem}[$\QMA(a^*,b) = \QMA_1$]
	\label{qma1theorem}
	Let $V$ be a verifier circuit for an instance of a language in $\QMA$, with promise bounds $a$ and $b$,
	and let $p_{max}$ be its highest acceptance probability $p_{max}$ in the `yes' case.
	When $\varphi_{\min} = \frac{1}{\pi} \arccos \sqrt{p_{max}}$
	can be expressed exactly by $t=polylog(|x|)$ bits, there exists a circuit $V^{*}$
	on $n+t$ qubits and $poly(n)$ ancillae whose acceptance probability in the `yes' case is exactly 1. 
	Said alternatively, $\QMA_1$ is equal to this subclass of $\QMA$.
\end{theorem}

\noindent Note that besides the witness, the new verifier circuit $V^{*}$ needs to receive the description of $a^*$. Also note that obtaining a result of this type
is possible only because of the coherent phase-estimation processing. Marriott and Watrous' $\QMA$ amplification scheme can not be used to amplify the probabilities to $1$ exactly, because of the inherent probabilistic nature of the measurements.


The paper is organized as follows. 
First, in Section \ref{projectorsection} we establish 
two important geometrical facts about two projectors, 
required for the analysis of the amplification protocols. 
In Section \ref{MWsection} we revisit 
the witness-preserving $\QMA$ amplification protocol of Marriott and Watrous \cite{MWqma}.
Then in Section \ref{phasesection} we give a faster $\QMA$ amplification method
utilizing phase estimation instead of classical final processing. 
Second, in Section \ref{qma1section} we show that in a special case when the 
maximum acceptance probability of the verifier circuit has a particular description, 
we can amplify it to 1 exactly.
Finally, in Section \ref{solvingsection}, inspired by our fast QMA amplification,
we simplify the method for preparing $\QMA$ witnesses by Poulin and Wocjan \cite{PoulinWocjan},
based on filter states.

We conclude with a discussion of our results in Section \ref{conclusions}.
Finally, Appendix \ref{jordansection} contains a new proof of the relationship of the eigenvalue gap of a classical random walk and the phase gap of a quantum walk, based on Jordan's lemma.


\section{Fast QMA Amplification}
\label{fastqma}

We arrive at our new QMA amplification scheme, based on phase estimating a certain unitary operator,
in several steps. First, with the help of Jordan's lemma, we look at how two projectors
act in a Hilbert space. This helps us to understand the Marriott and Watrous \cite{MWqma} amplification
scheme in Section \ref{MWsection}, based on alternative projective measurements.
Second, Jordan's lemma, together with a connection from quantum walks, leads us
to realize that a product of the reflections about the supports of the projectors
acts as a rotation within certain subspaces of the Hilbert space. Instead of projective measurements,
we thus base our amplification scheme of Section \ref{phasesection} on the phase estimation
of these rotations. This leads to a speedup in the number of the required evaluations of the
verifier circuit.


\subsection{Facts about Two Projectors and Two Reflections}
\label{projectorsection}

In this geometrically-focused section, we establish several facts required for 
understanding both Marriott and Watrous' witness preserving $\QMA$ amplification 
of Section \ref{MWsection}, and our new method in Section \ref{phasesection}.

First, consider two projectors $\Pi_0$ and $\Pi_1$ on the Hilbert space $\cH$.
Jordan's lemma \cite{Jordan75} (see also Appendix \ref{jordansection}) tells us that given two projectors, we can decompose our Hilbert space into 
\begin{enumerate}
	\item two-dimensional subspaces $S_i$ invariant under $\Pi_{0}$ and $\Pi_{1}$, and
	\item one-dimensional subspaces $T_j$, on which $\Pi_{0} \Pi_{1}$
				is an identity or a zero-rank projector. 
\end{enumerate}
Let us focus on the more interesting two-dimensional subspaces.
For each of them, we can choose a basis $\{\ket{v_i}, \kets{v_i^{\perp}}\}$, obeying
\begin{eqnarray}
	\Pi_0 \ket{v_i} = \ket{v_i}, \quad  
	\Pi_0 \kets{v_i^\perp} = 0. \label{vbasis}
\end{eqnarray}
We can also choose another basis $\{\ket{w_i}, \kets{w_i^{\perp}}\}$ for $S_i$ 
from the eigenvectors of $\Pi_1$, obeying
\begin{eqnarray}
	\Pi_1 \ket{w_i} = \ket{w_i},
	\quad
	\Pi_1 \kets{w_i^\perp} = 0. \label{wbasis}
\end{eqnarray}
We can also make a choice of phases so that $\braket{v_i}{w_i}$ is a real positive number,
and define the {\em principal angle} 
\begin{eqnarray}
	\varphi_i = \frac{1}{\pi}\arccos(|\braket{v_i}{w_i}|). 
	\label{principal}
\end{eqnarray}
Let $p_i$ be the expectation value of $\Pi_1$ in the state $\ket{v_i}$.
In Sections \ref{MWsection} and \ref{phasesection}, $p_i$ will be 
the probability that the verifier circuit $V$ accepts the state $\ket{v_i}$.
The projector $\Pi_1$ restricted to the subspace $S_i$ is $\ket{w_i}\bra{w_i}$, so we have
\begin{eqnarray}
	p_i = \bra{v_i} \Pi_1 \ket{v_i} = \braket{v_i}{w_i}\braket{w_i}{v_i} = \cos^2 (\pi \varphi_i).
	\label{paccept}
\end{eqnarray}
Using $p_i$, we now express
\begin{eqnarray}
	\ket{v_i} &=& \sqrt{p_i} \ket{w_i} + \sqrt{1-p_i} \kets{w_i^\perp}, \label{crossprobabilities}\\
	\ket{w_i} &=& \sqrt{p_i} \ket{v_i} + \sqrt{1-p_i} \kets{v_i^\perp}. \nonumber
\end{eqnarray}
This symmetrical relationship is an essential element in Marriott and Watrous' 
witness-preserving $\QMA$ amplification scheme in Section \ref{MWsection}.

Second, let us look at the reflections about the supports of $\Pi_0$ and $\Pi_1$:
\begin{eqnarray}
    F_0 &=& 2\Pi_0 - \ii, \label{reflect}\\
    F_1 &=& 2\Pi_1 - \ii. \nonumber
\end{eqnarray}
Using the decomposition of the Hilbert space from Jordan's lemma, we can prove 
the following theorem about the eigenvalues of the product of two reflections:
\begin{theorem}
	\label{producteigs} 
	Let $\Pi_0$ and $\Pi_1$ be projectors and $F_0$ and $F_1$ the reflections about their supports. 
	The Hilbert space can be decomposed into two-dimensional 
	and one-dimensional subspaces invariant under $\Pi_0$ and $\Pi_1$.
	The unitary operator $F_1 F_0 = (2 \Pi_1 - \ii)(2 \Pi_0 - \ii)$ 
	is a rotation $e^{i 2\pi\varphi_i \hat{\sigma}_y}$ with 
	eigenvalues $e^{\pm i 2 \pi \varphi_i }$ with $0<\varphi_i <\frac{1}{2}$ 
	corresponding to eigenvectors
	\begin{eqnarray}
		\kets{\varphi_i^+} = \frac{1}{\sqrt{2}} 
				\left( \kets{v_i} + i \kets{v_i^\perp}\right), \qquad
		\kets{\varphi_i^-} = \frac{1}{\sqrt{2}} 
				\left( \kets{v_i} - i \kets{v_i^\perp}\right), 
		\label{roteigs}
	\end{eqnarray}
	in the two-dimensional invariant subspaces $S_i$, 
	and eigenvalues $\pm 1$ in the one-dimensional invariant subspaces $T_j$.
\end{theorem}
\noindent 
Thus, within each two-dimensional subspace $S_i$, 
the product of two reflections $F_1 F_0$ is a rotation
by $2\pi\varphi_i$, with $\varphi_i$ given in \eqref{principal}. 
We rely on this fact in our fast $\QMA$ amplification scheme in Section \ref{phasesection}.

This theorem, proved by Szegedy \cite{Szegedy} is a base component in many results
about quantum walks. We also point the interested reader to a new proof of Theorem \ref{producteigs} 
using Jordan's lemma in Appendix \ref{jordansection}.


\subsection{Witness-preserving QMA Amplification}
\label{MWsection}

 \begin{figure}
 	\begin{center}
 	\includegraphics[width=2.5in]{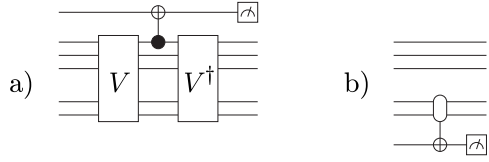} 
 	\end{center}
 	\caption{The circuits for measuring a) $\{\Pi_1,\ii-\Pi_1\}$ and b) $\{\Pi_0,\ii-\Pi_0\}$ in Marriott and Watrous' $\QMA$ amplification scheme.}
 	\label{figureMW}
 \end{figure}
We can now look at the $\QMA$ amplification scheme of Marriott and Watrous \cite{MWqma} 
which we want to speed up. 
Consider now the two projectors:
\begin{eqnarray}
    \Pi_0 &=& \ket{0}\bra{0}_{anc}, \label{twoP} \\
    \Pi_1 &=& V^\dagger \ket{1}\bra{1}_{out} V, \nonumber
\end{eqnarray}
with $\Pi_0$ projecting onto zeros on the ancilla qubits, and $\Pi_1$ projecting 
on the states that are accepted by the original verifier circuit. Note that the support
of $\Pi_1$ can contain states with nonzero ancillae. 
The promise bounds of the original circuit $V$ are $a$ and $b$.
The amplification procedure with the strong promise bounds \eqref{strongpromise} is the following:
\begin{enumerate}
	\item Combine a single input witness of length $n$ with fresh ancilla qubits.
	\item Perform a sequence of $N=\frac{cr}{(a-b)^2}$ alternating measurements 
			of $\{\Pi_0,\ii-\Pi_0\}$ and $\{\Pi_1,\ii-\Pi_1\}$ (see Figure \ref{figureMW}).
	\item Inspect the sequence of results, and count the number of times 
			when two consecutive results differ\footnote{To allow this number to reach $N$, attach a zero at the start of the sequence of measurement results.}.
			When this number is smaller than $N\frac{a+b}{2}$, 
			output `yes', otherwise output `no'. 
\end{enumerate}

Let us sketch why this scheme works. Recall the decomposition of the Hilbert space into 2-dimensional and 1-dimensional subspaces in Section \ref{projectorsection} and focus on the 2D subspaces $S_i$.
We can choose two bases for each subspace $S_i$ as in \eqref{vbasis} and \eqref{wbasis}.
Our $\Pi_0$ is such that the states $\ket{v_i}$ \eqref{vbasis} have the ancillae in the state zero, i.e. $\ket{v_i}=\ket{\psi_i}\ket{0}$. 
With our choice of $\Pi_1$, 
\begin{eqnarray}
	p_i = \bra{v_i}\Pi_1\ket{v_i} 
	= \bra{0}\bra{\psi_i}V^\dagger \ket{1}\bra{1}_{out} V \ket{\psi_i}\ket{0},
	\label{pmeaning}
\end{eqnarray}
defined in \eqref{paccept} is the probability that the original 
circuit $V$ accepts the witness $\ket{\psi_i}$.

We now choose some $\ket{v_i}$ as the initial state and perform the sequence of alternating measurements of $\Pi_0$ and $\Pi_1$, obtaining a bit string.
Throughout this procedure, the state of the system stays within the original 2-dimensional subspace $S_i$. 
Moreover, the identities in \eqref{crossprobabilities} imply that the probability of obtaining two consecutive 0's or 1's 
(projecting onto $\ket{w_i}$ from $\ket{v_i}$, etc.) is $p_i$, while the probability of obtaining $10$ or $01$ is $(1-p_i)$. The probability of a given sequence of measurement results is then 
\begin{eqnarray}
	p_{seq} = (1-p_i)^{z} p^{N-z},
\end{eqnarray}
where $z$ is the number of times consecutive measurement results differ.
When getting $z<\frac{N(a+b)}{2}$, we output `yes', and for $z>\frac{N(a+b)}{2}$
we output `no'. Marriott and Watrous use arguments based on Chernoff bounds to
show that these answers have exponentially good confidence bounds \eqref{strongpromise}. 
What remains is to show that in the `no' case, a superposition of the states $\ket{v_i}$ for different $i$'s will not help Merlin to fool Arthur.

To summarize, the scheme consists of $N$ measurements \eqref{Nkit}, and one half of them involves evaluating the circuit $V$ and its conjugate. Note that the processing of the sequence of results is based on classical statistical methods. In the next Section we show that using a natively quantum algorithm (phase estimation or filter states) results in faster amplification.


\subsection{Fast QMA Amplification Based on Phase Estimation}
\label{phasesection}

In our new amplification procedure, 
we utilize the same pair of projectors $\Pi_0$ and $\Pi_1$ \eqref{twoP} as Marriott and Watrous.
However, instead of measuring them directly, we build our circuit
using the reflections $F_0$ and $F_1$ \eqref{reflect} about their supports.
Theorem \ref{producteigs} 
tells us that within the two-dimensional invariant subspaces 
introduced in Section \ref{projectorsection},
the product $F_1 F_0$ is a rotation by an angle related 
to an acceptance probability for the original circuit $V$. 
This turns the problem of accepting or rejecting witnesses for the circuit $V$ into the problem of determining the properties of the rotation coming from the operator $F_1 F_0$. 
We will show that a small rotation corresponds 
to a high acceptance probability for the circuit $V$ and vice versa. 
We thus do phase estimation on the operator $F_1 F_0$, 
accepting or rejecting the witness depending on the phase we obtain.
Finally, to boost the probability of success, we concatenate several such phase estimation procedures,
obtaining the desired strong promise bounds \eqref{strongpromise}.

First, let us look at a `yes' case, and show that Merlin can convince us about it.
Just as in Section \ref{MWsection}, he chooses a witness $\ket{\psi_i}$,
which we combine with fresh ancillae to get
\begin{eqnarray}
	\ket{v_i} = \ket{\psi_i}\ket{0},
	\label{vi}
\end{eqnarray}
an eigenvector of $\Pi_0$, belonging to one of the two-dimensional invariant subspaces $S_i$ 
we introduced in Section \ref{projectorsection}.
From Theorem \ref{producteigs}, we know that 
within this subspace, $F_1 F_0$ has the form $e^{i2\pi\varphi_i \hat{\sigma}_y}$ 
and its eigenvalues are $e^{\pm i2\pi \varphi_i}$,
with $\varphi_i$ given by \eqref{principal}. 
Recalling \eqref{roteigs}, we can express $\ket{v_i}$ in terms of the 
eigenvectors $\kets{\varphi_i^\pm}$ of $F_1 F_0$ as
\begin{eqnarray}
	\ket{v_i}  
	= \frac{1}{2}
								\left( 
									\ket{v_i} + i \kets{v_i^{\perp}}
								\right)
 		+ \frac{1}{2} 
								\left( 
									\ket{v_i} - i \kets{v_i^{\perp}}
								\right)			
	 = \frac{1}{\sqrt{2}}\left(\kets{\varphi_i^+} + \kets{\varphi_i^-}\right).
	 \label{eigvec}
\end{eqnarray}
Merlin chooses his witness $\ket{\psi_i}$ so that the corresponding phase 
$\varphi_i$ is the smallest. According to \eqref{pmeaning} and \eqref{paccept}, 
such $\varphi_i$ corresponds to picking the witness $\ket{\psi_i}$ with the 
largest possible acceptance probability $p_i$.
Because we are talking about the `yes' case, this $\varphi_i^{(yes)}$ is upper bounded by
\begin{eqnarray}
	\varphi_i^{(yes)} = \frac{1}{\pi}\arccos \sqrt{p_i} \leq \frac{1}{\pi}\arccos \sqrt{a} = \varphi_a, 
	\label{thetaa}
\end{eqnarray}
where $a$ is the guaranteed acceptance probability for some witness in the definition of $\QMA$.
Similarly, for a `no' case, the smallest possible $\varphi_i^{(no)}$ is lower bounded by
\begin{eqnarray}
	\varphi_i^{(no)} = \frac{1}{\pi}\arccos \sqrt{p_i} \geq \frac{1}{\pi}\arccos \sqrt{b} = \varphi_b.
	\label{thetab}
\end{eqnarray}
Our approach will be to measure the rotation phase for the state $\ket{v_i}$ \eqref{eigvec},
and to resolve whether it is less than $\varphi_a$, or larger than $\varphi_b$.
Note that the phase estimation for the operator $F_1 F_0$ on the state \eqref{eigvec}
has two possible results $\pm \varphi_i$. However, because of the form of \eqref{paccept}, we are only interested in the absolute value $|\varphi_i|$. We can recognize we got an estimate of the negative phase $-\varphi_i$ when we measure $\varphi'\in \left[\frac{1}{2},1\right]$, and we then use the value $1-\varphi' \in \left[0,\frac{1}{2}\right]$ instead of it.

We get the state \eqref{eigvec} on the input, and our phase estimation outputs some value $\varphi'$.
To be convinced that $|\varphi'|$ is smaller than $\varphi_a$ or 
greater than $\varphi_b$, we need a precision guarantee 
\begin{eqnarray}
	\delta\varphi = \frac{\varphi_b-\varphi_a}{2}.
	\label{precision}
\end{eqnarray} 
Using the Taylor expansion of $\arccos{x}$ around $x=0$, we bound
\begin{eqnarray}
 \varphi_b-\varphi_a &=& \frac{1}{\pi}\left(\arccos{\sqrt{b}}-\arccos{\sqrt{a}}\right) \nonumber \\
 	&=& \frac{1}{\pi} \sum_{n=0}^{+\infty}
  \frac {(2 n)!} {2^{2n} (n!)^2}  \frac 1 {2 n+1} 
  	\left[
  			\left( \sqrt{a} \right)^{2n+1} 
  			-\left( \sqrt{b} \right)^{2n+1}
  	\right], 
   \nonumber\\
 &\geq& \frac{\sqrt{a}-\sqrt{b}}{\pi}  = \frac{a-b}{\pi\left(\sqrt{a}+\sqrt{b}\right)}, \label{upper}
\end{eqnarray} 
as all the coefficients in the sum are positive and we have $a>b$,
so we can lower bound the expression by the $n=0$ term.
Therefore, we need 
\begin{eqnarray}
		n = - 1 + \log \frac{1}{\delta \varphi}  
		= - 1 + \log \frac{2}{\varphi_b-\varphi_a}
		= \log \frac{1}{\varphi_b-\varphi_a}
		\leq \log\frac{\pi\left(\sqrt{a}+\sqrt{b}\right)}{a-b}
		\label{nprecision}
\end{eqnarray}
bits of precision for our phase estimation.
According to \cite{NCbook}, a phase estimation algorithm
precise to $n$ bits, with failure probability $\ep_{pe}$
requires
\begin{eqnarray}
    t = n + \log \left[ 2 + 1/(2\ep_{pe})\right]
	\label{tqubits}
\end{eqnarray}
ancilla qubits. The number of times we need to perform
a controlled-$(F_1F_0)$ operation is then
\begin{eqnarray}
	N_1 = 2^t - 1 
	\leq 2^n \left( 2 + \frac{1}{2\ep_{pe}}\right)
	\leq \frac{10\pi\left(\sqrt{a}+\sqrt{b}\right)}{a-b},	
	\label{phase1}
\end{eqnarray}
using \eqref{nprecision} and choosing $\ep_{pe}\leq\frac{1}{16}$
as an upper bound on the failure probability.

 \begin{figure}
 	\begin{center}
 	\includegraphics[width=2in]{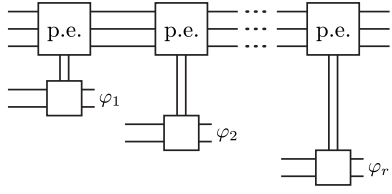} 
 	\end{center}
 	\caption{The scheme for concatenating $r$ phase estimations of $F_1 F_0$.}
 	\label{figurephase}
 \end{figure}
We want our procedure to work with exponentially small failure probability. 
To achieve this, we concatenate $r$ phase estimation circuits as in Figure 
\ref{figurephase} and take the median of the $r$ results (recall that
we use $|\varphi| \in \left[0,\frac{1}{2}\right]$ here). 
Because of the following lemma\footnote{Our median lemma is a 
variant of the {\em powering lemma} \cite{Valiant}.},
the median phase is a good estimate of the actual phase $\varphi_i$, 
with high probability.
\begin{lemma}[Median lemma]\label{lem:med}
	Consider a sequence $\{\varphi'_k\}$ for $k=1,\dots,r$.
	Let the probability that $\varphi'_k$ does not belong to an interval $LR=(\varphi_L,\varphi_R)$ 
	be $\ep$, with $\ep<\frac{1}{2}$.
	Then the probability that the median of $\{\varphi'_k\}$ falls out of $LR$ is
	bounded by
	$p_{fail} \leq \frac{1}{2} \left(2\sqrt{\ep(1-\ep)}\right)^{r}$.
\end{lemma}
\begin{proof}
	The only way the median could fall out of the interval $LR$ is by having more than half of the 
	datapoints $\varphi'_k$ falling out of it. This probability is bounded by
	\begin{eqnarray}
		p_{fail} &\leq& \sum_{k=\frac{r}{2}}^{r} \ep^k (1-\ep)^{r-k} {r\choose{k}} 
			= (1-\ep)^r \sum_{k=\frac{r}{2}}^{r} \left(\frac{\ep}{1-\ep}\right)^k {r\choose{k}} \nonumber \\
			&\leq& (1-\ep)^r \left(\frac{\ep}{1-\ep}\right)^{\frac{r}{2}} 
					\sum_{k=\frac{r}{2}}^{r} {r\choose{k}} 
			=  (1-\ep)^r \left(\frac{\ep}{1-\ep}\right)^{\frac{r}{2}} 2^{r-1} \\
			&=&  \frac{1}{2} \left(2\sqrt{\ep(1-\ep)}\right)^{r}, \nonumber
	\end{eqnarray}
	where we used the fact that $\ep<\frac{1}{2}$, so that $\frac{\ep}{1-\ep}<1$ and 
	$2\sqrt{\ep(1-\ep)}<1$.
\end{proof}

\noindent 
Recall that the precision \eqref{precision} we demanded for $\varphi'$ 
is such that when we get a phase smaller than 
$\frac{\varphi_a + \varphi_b}{2}$,
we conclude that we have a `yes' case.
Lemma \ref{lem:med} tells us that for $\ep_{pe}=\frac{1}{16}$, 
the probability of our median scheme failing (producing a bad estimate of the phase) 
is bounded from above by
\begin{eqnarray}
	p_{fail} \leq \frac{1}{2} \left(2\sqrt{\ep_{pe}(1-\ep_{pe})}\right)^{r}
		\leq 2^{-r}.
	\label{pfail}
\end{eqnarray} 
This gives us the first of the strong promise bounds \eqref{strongpromise},
with
\begin{eqnarray}
	N' = rN_1 = \frac{10\pi\left(\sqrt{a}+\sqrt{b}\right)r}{a-b}.	
	\label{phase1r}
\end{eqnarray}
uses of the verifier circuit and its inverse.
In the worst case (when the phases are hardest to tell apart), we have $a\approx b\approx \frac{1}{2}$, 
giving us $N' = \frac{5\pi r}{2(a-b)}$.
We can compare this to 
\eqref{Nkit}, where the actual value of the constant is $c=2$, 
and see that our method is better already for $a-b \leq \frac{1}{4}$.
Of course, our motivation to start thinking about a new $\QMA$ amplification method was
the case when $a-b$ is tiny. 

For the `no' case, the analysis is the same as above, if $\ket{\psi}\ket{0}$ (the witness combined with fresh ancillae) belongs to one of the two-dimensional invariant subspaces.
The smallest median phase we can get corresponds to the largest possible acceptance probability, 
which is lower than $b$. The probability of obtaining a median of the $r$ the phase estimations below $\frac{\varphi_a + \varphi_b}{2}$ (i.e. Arthur being fooled) is again upper bounded by \eqref{pfail}.
On the other hand, what if we get a superposition $\sum_{i} v_i \ket{v_i}$ 
with $\ket{v_i}$'s from different two-dimensional invariant subspaces on the input?
The final phase measurement on a superposition like 
$\sum_{i} v_i \ket{v_i} \ket{\varphi_i}$
gives us a classical mixture of results $\varphi_i$, weighted by $v_i^2$. Therefore, it's always better for Merlin to choose a single $i$ with the smallest possible $\varphi_i^{(no)}$ if he wants to have a chance of fooling us.
Nevertheless the probability to measure a phase smaller than $\frac{\varphi_a + \varphi_b}{2}$
in a `no' case is upper bounded by $\ep_{pe}$ for any $\ket{\alpha_i}$. 
Recalling what we did above, several concatenated phase estimations
allow us to detect that $\varphi' > \frac{\varphi_a+\varphi_b}{2}$ with high probability.
We then conclude that Merlin is just trying to fool us.



\section{Is QMA with one-sided error equal to QMA?}
\label{qma1section}

Merlin knows all about our verifier circuit. This means he also knows the
division of the Hilbert space corresponding to the two projectors $\Pi_0$ and $\Pi_1$
described in Section \ref{projectorsection},
and the bases $\{\ket{v_i},\kets{v_i^\perp}\}$ of the two-dimensional invariant subspaces.
Each of the vectors $\ket{v_i}=\ket{\psi_i}\ket{0}$ is a combination of the eigenvectors \eqref{eigvec}
of the product of two reflections $F_1 F_0$ \eqref{reflect}, with eigenvalues
$e^{\pm i 2\pi \varphi_i}$. Through \eqref{paccept} and \eqref{pmeaning},
the phase $\varphi_i$ is related to the probability
that the circuit $V$ accepts the witness $\ket{\psi_i}$.
If Merlin sent us $\ket{v_i}$ and told us what the phase $\varphi_i$ is, we could verify his claim.
However, the phase estimation circuit is not always perfect. Still, it works exactly,
if the phase we are estimating has an $n$-bit binary expansion. 
Thus, we could do the following:
\begin{enumerate}
	\item Merlin sends us the state $\ket{\psi_i}_{witness} \otimes \ket{\varphi_i}_{n}$,
				with $\ket{\varphi_i}_{n}$ an exact $n$-bit binary encoding of $\varphi_i$.
	\item Measure $\ket{\varphi_i}_{n}$ in the $z$ basis, obtaining the number $\varphi_i$.
				Continue if $p_i = \cos^2(\pi \varphi_i) \geq a$.
	\item Run the fast QMA amplification scheme. When the witness is $\ket{\psi_i}$, 
				the resulting phases we measure must all be equal to $\varphi_i$ or $-\varphi_i$,
				because the $n$-bit phase estimation should work perfectly 
				for $\varphi_i$ with an exact $n$-bit binary expansion.
	\item Only when the result agrees with the $\varphi_i$ which Merlin has sent, 
				we are convinced that the answer to the yes/no question is `yes'.
\end{enumerate}			
For the `yes' cases, the acceptance probability is exactly 1 when Merlin does everything right. 
On the other hand, if the answer is `no', the acceptance probability of the QMA amplification 
scheme is already exponentially small, as shown in Section \ref{phasesection}.
This finishes the proof of our Theorem \ref{qma1theorem}: for a special class of verifier circuits, 
$\QMA_\varphi$ is equal to $\QMA_1$.

Today, we do not know which circuits $V$ have the nice property described above. 
Nevertheless, there is a possibility that circuits without it 
can be (easily) slightly modified to have it.
Merlin could send us a classical hint about the modification, and we would do it, 
keeping the properties of the circuit required in the definition of $\QMA$.
If this were possible, $\QMA_1 = \QMA$.


\section{Preparing Witnesses for QMA}
\label{solvingsection}

Given a verifier circuit $V$ for a problem in $\QMA$, 
what could we do to actually prepare a witness which the circuit accepts?
Poulin and Wocjan investigated this \cite{PoulinWocjan} 
together with the problem of preparing ground states of local Hamiltonians.
This question is not simple. The first idea would be to do a basic Grover search \cite{Grover} for the 
states in the support of the projector $\Pi_1$ (the states on which the circuit $V$ outputs $1$).
However, this works only when $\Pi_1$ commutes with $\Pi_0$ (the projector on zeros on the ancillae).
When $[\Pi_1,\Pi_0]\neq 0$, the ancilla part of the states we get from Grover searching will very likely be nonzero, and the method does not produce a proper witness of the form $\ket{\psi_i}\ket{0}$. 
Poulin and Wocjan found a way for preparing the witness in general. First, they run the witness-preserving $\QMA$ amplification scheme of Marriott and Watrous \cite{MWqma} (see Section \ref{MWsection}) backwards, and then do Grover search for the part of the state with zero ancillae.
We simplify their method, showing how to search for $\QMA$
witnesses using a reverse of our fast $\QMA$ amplification, in a much smaller system with easier initialization.
This also unifies their approaches to ground state and $\QMA$ witness preparation.

Let us sketch the new {\em filter state} method.
After the phase estimation (before the final Fourier transform) in our fast $\QMA$ amplification, we expect to have the state
\begin{eqnarray}
	\ket{\varphi_i^+}_1 \underbrace{\frac{1}{2^n}\sum_{k=0}^{2^{n}-1} e^{2\pi i n \varphi_i} \ket{n}_2}_{
		\ket{f_{\varphi_i}}
	}
	\label{afterphase}
\end{eqnarray}
in superposition with a corresponding $\ket{\varphi_i^-}_1$ part.
The state $\ket{f_{\varphi_i}}$ is the {\em filter state} for phase $\varphi_i$.
Imagine now that we ran the phase estimation backwards, starting in \eqref{afterphase}. We would obtain $\ket{\varphi_i^+}_1\ket{0}_2$. 
What would happen if we ran phase estimation backwards on 
\begin{eqnarray}
		\ket{\alpha}_1\ket{f_{\varphi_i}}_2,
\end{eqnarray}
where $\ket{\alpha}$ is a random state and $\ket{f_{\varphi_i}}$ is the filter state?
We would get 
\begin{eqnarray}
		\underbrace{\braket{\varphi_i^+}{\alpha}}_{c_{\varphi_i}}\kets{\varphi_i^+}_1\ket{0}_2 
		+ c_{\varphi_i}^{\perp} \kets{\varphi_i^\perp}_1\ket{0}_2 + c_{\beta} \ket{\beta}_{12},
\end{eqnarray}
where the second register of $\ket{\beta}_{12}$ is nonzero.
If we created the filter state $\ket{f_{\varphi_i}}$ with a large $n$, 
the coefficient $c_{\varphi_i}^{\perp}$ must be small (with high probability). 
%
%
The last step is then to amplify the coefficient $c_{\varphi_i}$ by amplitude amplification for the 
states with zeros on the second register -- the phase estimation qubits. After we do this, we will obtain the state
\begin{eqnarray}
	\kets{\varphi_i^+}_1\ket{0}_{2},
\end{eqnarray}
with high probability, and from \eqref{eigvec} we know that 
\begin{eqnarray}
	\kets{\varphi_i^+} = \frac{1}{\sqrt{2}} \left( \kets{v_i} + \kets{v_i^\perp} \right)
	=
	\frac{1}{\sqrt{2}} \ket{\psi_i}\ket{0} 
	+ \frac{1}{\sqrt{2}} \kets{v_i^\perp}.
\end{eqnarray}
With probability $\frac{1}{2}$, projecting on the zeros of 
the ancillae then gives us the witness $\ket{\psi_i}$. In practice, we would have to repeat this 
probabilistic method many times, scanning a range of $\varphi$'s and verifying (with our fast $\QMA$ 
amplification method) whether we actually got a witness.

The filter state $\ket{f_{\varphi_i}}$ in our scheme sketched above is the Fourier transform of $\ket{\varphi_i}$.
In their first method for preparing ground states of many-body systems, Poulin and Wocjan 
use the same filter state, with the role of $\varphi_i$ played by the ground state energy.
Scanning a range of energies plays the same role as scanning through a range of phases in the
witness-preparation method. The analysis of the required number of phase estimation qubits $n$ 
to ensure that $c_{\varphi_i}^{\perp}$ is small, and the bounds on $c_{\varphi_i}$ 
thus works here as well.
We point the interested reader to Appendix C of \cite{PoulinWocjan} for the necessary details.

The original method for preparing witnesses in \cite{PoulinWocjan} differs from the method 
for preparing ground states of many-body systems. It uses {\em filter states} representing 
measurements outputs of the Marriott-Watrous scheme corresponding to a certain acceptance 
probability. For the filter states to work (see Appendix D in \cite{PoulinWocjan}), 
they need to introduce an extra register containing a large number of ancillae, 
scaling as $N$ in \eqref{Nkit}. The final Grover search in this method is over the space of these ancillae.

By using the reverse of our QMA amplification scheme, we achieve two things.
First, we unify the approaches to preparing ground states of many-body systems and QMA witnesses in \cite{PoulinWocjan}, by using the same type of filter states. 
Second, comparing \eqref{phase1r} and \eqref{Nkit} 
shows that we now need only $N' \propto~ (a-b) N$ ancilla qubits.
This greatly reduces the size of the final search space and speeds up the method.
Nevertheless, its running time obviously still stays exponential in $n$, the size of the problem.


\section{Discussion}
\label{conclusions}
First, we study the cost of amplifying the gap
between the acceptance probabilities in the `yes' and `no' cases for the
complexity class $\QMA$. Let $L$ be an arbitrary language in QMA. Given a
family $\{U_n\}$ of $\QMA(a,b)$-circuits accepting $L$, we show how to
construct a family $\{V_n\}$ of $\QMA(1-2^{-r},2^{-r})$-circuits accepting
$L$. Each of the circuits $V_n$ applies the original circuit $U_n$ or its inverse $U_n^\dagger$ at most $O\left(\frac{r}{a-b}\right)$
times. Thus, the complexity of our amplification method grows {\em linearly} in $\frac{1}{a-b}$.  This improves upon the
performance of the amplification method in \cite{MWqma} whose complexity grows {\em quadratically} in $\frac{1}{a-b}$. This quadratic speed-up is reminiscent of the speed-up in Grover's search algorithm and search algorithms employing quantum walks. This is not a coincidence. In fact, the intuition behind our amplification procedure is based on Szegedy's quantization of classical Markov chains -- quantum walks \cite{Szegedy}.  

To explain this intuition, let us first look at quantum walks from a point of view than slightly different  from the one usually taken in the literature. Roughly speaking, a quantum walk is derived from two projectors $P$ and $Q$ such that 
\begin{enumerate}
\item the unique state $|\pi\>$ with $\|P|\pi\>\|^2=1$ and $\|Q|\pi\>\|^2=1$ is the desired {\em quantum sample} of the stationary distribution of the corresponding classical walk,
\item all the other states $|\varphi\>$ with the property $\norm{P\ket{\varphi}}^2=1$ and orthogonal to $\ket{\pi}$ are necessarily contracted by $Q$, meaning that $\norm{Q\ket{\varphi}}^2 \leq 1-\delta$. 
\end{enumerate}
Here $\delta$ denotes the spectral gap of the corresponding classical walk.  To distinguish between $|\pi\>$ and $|\varphi\>$, we could alternatively measure the state according to the POVMs $\{P,\ii-P\}$ and $\{Q,\ii-Q\}$. When we have the state $\ket{\pi}$, we always obtain only the outcomes associated to $P$ and $Q$. On the other hand, for any of the states $\ket{\varphi}$, the probability of obtaining an outcome associated to $\ii-P$ or $\ii-Q$ is at least $\delta$. To obtain such an outcome with a constant probability, we have to make $O(\delta^{-1})$ measurements. Thus, the task of distinguishing between $\ket{\pi}$ and $\ket{\varphi}$ can be accomplished with complexity $O(\delta^{-1})$.

We can reduce this complexity with the help of the quantum walk $W=(2P-I)(2Q-I)$. The fact that $\norm{Q\ket{\pi}}^2 = 1$ translates into the fact that $\ket{\pi}$ is the unique eigenvector of $W$ with eigenvalue $1$ (and corresponding phase $0$). Moreover, the fact that $\norm{Q\ket{\varphi}}^2 \leq 1-\delta$ translates into the fact that $\ket{\varphi}$ is necessarily a superposition of eigenvectors of $W$ whose phases have absolute value greater than some $\Delta$. This {\em phase gap} $\Delta$ is related to the {\em eigenvalue gap} $\delta$ by a quadratic relation: $\Delta \ge \sqrt{\delta}$. This leads to the quantum speed-up, because we can now distinguish between the two cases by running phase estimation.
The required accuracy is $O(\Delta)$, and we can achieve this by invoking $W$ at most $O(\Delta^{-1}) = O(\delta^{-1/2})$ times.

The problem of verifying witnesses for $\QMA$ problems can be readily formulated with the help of two projectors.  $P$ is the projector onto all ancillae being in the state $|0\>$ and $Q$ is the projector onto the output qubit being in $|1\>$.  In the `yes' case, there is a state $|\psi\>$ with $\|P|\pi\>\|^2=1$ and $\|Q|\pi\>\|^2>a$ and in the `no' case we have $\|Q|\pi\>\|^2<b$ for all $|\psi\>$ with $\|P|\psi\>\|^2=1$. We could tell these two cases apart with constant probability by alternatively measuring the state according to the POVMs $\{P,\ii-P\}$ and $\{Q,\ii-Q\}$ at most $O\left(\frac{1}{(a-b)^2}\right)$ times.  In fact, this is exactly what the amplification procedure in \cite{MWqma} does.

Our fast $\QMA$ amplification extends the quadratic relation between the probability and phase gaps to the more general situation, where the acceptance probability is at least $a$ in the `yes' case and at most $b$ in the `no' case. Recall that the situation relevant to quantum walks corresponds to $a=1$ and $b=1-\delta$. In the general case ($a\leq 1$), we prove that the phases $\varphi_a$ and $\varphi_b$ corresponding to the probability bounds $a$ and $b$ satisfy the separation bound $\varphi_b-\varphi_a=\Omega(a-b)$.  By employing phase estimation, we can resolve these two cases by invoking $W$ at most $O\left(\frac{1}{a-b}\right)$ times.

Second, we study the complexity-theoretic question whether $\QMA$ is equal to $\QMA_1$. Based on our new amplification procedure, we show that in some special cases (when the largest possible acceptance probabilities $a^*$ in the `yes' cases satisfy a trigonometric identity), the acceptance probability in the `yes' case can be amplified to $1$. The idea is that in these special cases the probabilities $a^*$ translate into `nice' phases $\varphi_{a^*}$, which we can deterministically identify using phase estimation. 

In the future, we plan to examine whether we could exploit the quadratic relation between the probability and phase gaps in more general situations to obtain new faster quantum algorithms. We will also seek to determine ways of proving that the acceptance probability can be boosted to $1$ in new, less restrictive cases.


\section{Acknowledgments}
        D.~N. gratefully acknowledges support by 
        European Project QAP 2004-IST-FETPI-15848 and 
        by the Slovak Research and Development Agency under the contract No. APVV-0673-07.
        P.~W. gratefully acknowledges the support by NSF grants
        CCF-0726771 and CCF-0746600. Y.~Z. was partially supported by
        the NSF-China Grant-10605035. Part of this work
        was done while D.~N. was visiting University of Central Florida.


\appendix

\section{A proof of Theorem \ref{producteigs} using Jordan's lemma}
\label{jordansection}

Theorem \ref{producteigs} is an important result relevant to quantum
walks, and it was proved by Szegedy in \cite{Szegedy}.
We now prove it in a different way -- using Jordan's lemma.
Jordan's lemma has been recently used to analyze $\QMA$ amplification 
\cite{MWqma}, and here we show that it is useful for quantum walks as well. 
For a short proof of Jordan's lemma, see e.g. \cite{RegevNotes}.

\begin{lemma}[Jordan '75]
For any two Hermitian projectors $\Pi_1$ and $\Pi_0$, there exists an
orthogonal decomposition of the Hilbert space into one dimensional
and two dimensional subspaces that are invariant under both $\Pi_1$ and $\Pi_0$.
Moreover, inside each two-dimensional subspace, $\Pi_1$ and $\Pi_0$ are
rank-one projectors.
\end{lemma}

Consider now an $N$-dimensional Hilbert space $\cH$, a rank $r$ projector
$\Pi_1$ and a rank $s$ projector $\Pi_0$, with $1\le r, s \le N$. 
Jordan's lemma implies the existence of an orthonormal basis for the Hilbert space
$\cH$ which can be divided into five groups. 
\begin{enumerate}
\item Two-dimensional
		subspaces $S_i$ for $i=1,\dots,t$, invariant under both $\Pi_1$ and $\Pi_0$. Each
		subspace $S_i$ is spanned by the orthonormal eigenvectors $\ket{v_i}$ and $\kets{v_i^{\perp}}$
		of the projector $\Pi_1$, i.e. obeying 
		$\Pi_1 \ket{v_i} = \ket{v_i}$ and  $\Pi_1\kets{v_i^{\perp}} = 0$.
\item Four types of one dimensional subspaces $T^{(bc)}_i$, where $b,c\in\{0,1\}$. These 
		$N-2t$ subspaces are spanned by $\kets{v^{(bc)}_i}$, for $b,c\in\{0,1\}$,
		obeying
		\begin{eqnarray}
			\Pi_1 \kets{v^{(bc)}_i} = b \kets{v^{(bc)}_i}, \qquad
			\Pi_0 \kets{v^{(bc)}_i} = c \kets{v^{(bc)}_i}. 
			\label{onedim}
		\end{eqnarray}
\end{enumerate}

The 2D subspaces $S_i$ can be also spanned by the orthonormal
eigenvectors $|w_i\rangle, |w_i^\perp\rangle$ of the projector $\Pi_0$,
satisfying $\Pi_0 |w_i\rangle =|w_i\rangle$ and $\Pi_0 \kets{w_i^\perp} =0.$
Hence, we can recast the projectors $\Pi_1$ and $\Pi_0$ in the form 
\begin{eqnarray}
\label{basis_h}
 \Pi_1 &=& \sum_{i=1}^{t} \ket{v_i}\bra{v_i} 
 			+ \sum_{i} \kets{v^{(10)}_i}\bras{v^{(10)}_i} 
 			+ \sum_{i} \kets{v^{(11)}_i}\bras{v^{(11)}_i}, \\
 \Pi_0 &=& \sum_{i=1}^{t} \ket{w_i}\bra{w_i} 
 			+ \sum_{i} \kets{v^{(01)}_i}\bras{v^{(01)}_i} 
 			+ \sum_{i} \kets{v^{(11)}_i}\bras{v^{(11)}_i}.\nonumber
\end{eqnarray} 
in terms of our chosen orthonormal basis of the Hilbert space $\cH$.

We now rewrite Theorem \ref{producteigs} using the notation $\theta_i=\pi \varphi_i$,
and provide a new proof based on Jordan's lemma.
\begin{theorem3}
Consider two Hermitian projectors $\Pi_1$, $\Pi_0$ and the identity operator
$\ii$. The unitary operator $(2 \Pi_1-\ii)(2 \Pi_0-\ii)$ has eigenvalues $e^{\pm
i 2 \theta_i }$,  $0<\theta_i <\frac \pi 2$  in the two-dimensional
subspaces $S_i$ invariant under $\Pi_1$ and $\Pi_0$, and it has eigenvalues
$\pm 1$ in the one-dimensional subspaces invariant under $\Pi_1$ and $\Pi_0$.
\end{theorem3}

\begin{proof}
With a suitable choice of the orthonormal eigenvectors
$\{|w_i\rangle, |w_i^\perp\rangle \}$ of the projector $\Pi_0$ 
and defining\footnote{Note that in \eqref{principal}, we chose to use
the rescaled $\varphi_i=\frac{\theta_i}{\pi}$ instead, to unify the notation with phase estimation.}
 the principal angle $\theta_i$
\begin{eqnarray}
		\cos\theta_i=|\langle v_i | w_i \rangle|,
\end{eqnarray}
we can write the transformation law
\begin{eqnarray}
	\kets{v_i} &=& \phantom{-} \cos \theta_i \kets{w_i} + \sin \theta_i \kets{w_i^\perp}, \\
	\kets{v_i^\perp} &=& - \sin \theta_i \kets{w_i} + \cos \theta_i \kets{w_i^\perp}.
\end{eqnarray}
between the two bases as
\begin{eqnarray}
	\left(\begin{array}{l}
			\kets{w_i}\\
			\kets{w_i^\perp}
		\end{array}\right)
	= U   
		\left(\begin{array}{l}
 			\kets{v_i} \\
 			\kets{v_i^\perp}
		\end{array}\right),
\end{eqnarray}
where the transformation
\begin{eqnarray}
	U = e^{i \theta_i \hat{\sigma}_y} = \ii \cos \theta_i + i \hat{\sigma}_y \sin\theta_i
\end{eqnarray}
is a unitary rotation. 
Therefore, $(2\Pi_1-\ii)(2\Pi_0-\ii)$ expressed in the basis $\{\kets{v_i},\kets{v_i^\perp}\}$ is
\begin{eqnarray}
	(2\Pi_1-\ii)(2\Pi_0-\ii) 
	&=&
		\left[\begin{array}{rr}
				1 & 0\\ 
				0 & -1
		\end{array}\right]
		\underbrace{U
		\left[\begin{array}{rr}
				1 & 0\\ 
				0 & -1
		\end{array}\right]
		U^{\dagger}}_{2\Pi_0-\ii} \\
	&=& \hat{\sigma}_z e^{i \theta_i \hat{\sigma}_y} \hat{\sigma}_z e^{-i \theta_i \hat{\sigma}_y} 
	= \hat{\sigma}_z \hat{\sigma}_z e^{-i \theta_i \hat{\sigma}_y} e^{-i \theta_i \hat{\sigma}_y} \\
	&=& e^{-i (2\theta_i) \hat{\sigma}_y},  \label{PPform}
\end{eqnarray}
where we used the anticommutation relationship $\hat{\sigma}_z \hat{\sigma}_y + \hat{\sigma}_y \hat{\sigma}_z = 0$.
In each two-dimensional subspace $S_i$, the operator $(2\Pi_1-\ii)(2\Pi_0-\ii)$ is thus a rotation 
by $2\theta_i$. Moreover, its form is \eqref{PPform} also in the basis
$\{\kets{w_i},\kets{w_i^\perp}\}$, because $U^\dagger e^{-i (2\theta_i) \hat{\sigma}_y} U
=e^{-i (2\theta_i) \hat{\sigma}_y}$.

To conclude the proof, let us look at the action of $(2\Pi_1-\ii)(2\Pi_0-\ii)$ 
on the one-dimensional invariant subspaces.
Using \eqref{onedim}, we find that it acts as an identity on the subspaces $T^{(00)}_i$ and $T^{(11)}_i$
and as $-\ii$ on the subspaces $T^{(01)}_i$ and $T^{(10)}_i$.
\end{proof}


\end{document}